\documentclass[10pt,journal,onecolumn,letterpaper]{IEEEtran}

\IEEEoverridecommandlockouts
\usepackage{cite}
\usepackage{amsmath,amssymb,amsfonts}
\usepackage{algorithmic}
\usepackage{graphicx}
\usepackage{textcomp}
\usepackage{xcolor}
\def\BibTeX{{\rm B\kern-.05em{\sc i\kern-.025em b}\kern-.08em
    T\kern-.1667em\lower.7ex\hbox{E}\kern-.125emX}}

\usepackage[utf8]{inputenc} 
\usepackage[T1]{fontenc}    
\usepackage{hyperref}       
\usepackage{url}            
\usepackage{booktabs}       
\usepackage{amsfonts}       
\usepackage{nicefrac}       
\usepackage{microtype}      
\usepackage{xcolor}         

\usepackage{microtype}
\usepackage{graphicx}
\usepackage{booktabs} 
\usepackage{enumerate}

\usepackage{hyperref}
\usepackage{url}
\usepackage{amsthm} 
\newtheorem{theorem}{Theorem}

\newtheorem{lemma}{Lemma}
\newtheorem*{lemma*}{Lemma}
\newtheorem*{claim*}{Claim}

\newtheorem{problem}{Problem}
\theoremstyle{definition}
\newtheorem{definition}{Definition}

\newtheorem{assumption}{Assumption}
\theoremstyle{remark}

\usepackage{mathtools, nccmath}
\usepackage{multicol,amsmath,bm}
\usepackage{booktabs} 
\usepackage{subcaption}
\usepackage[linesnumbered,ruled]{algorithm2e}

\usepackage{graphicx}

\newcommand{\alg}{{BDTD}~}

\usepackage{xcolor}

\usepackage{tikz}
\usepackage{anyfontsize}
\usetikzlibrary{arrows,matrix,positioning,decorations.pathreplacing,calc}
\usepackage{multirow}
\usepackage{hhline}
\allowdisplaybreaks

\usepackage{hyperref}
\newcommand{\TwoNorm}[1]{\left\|#1\right\|}
\DeclareMathOperator*{\argmax}{arg\,max}

\usepackage{amsmath}
\usepackage{amssymb}
\usepackage{mathtools}
\usepackage{amsthm}

\usepackage[capitalize,noabbrev]{cleveref}

\theoremstyle{plain}

\theoremstyle{definition}

\theoremstyle{remark}

\usepackage[textsize=tiny]{todonotes}

\begin{document}

\title{On the Hardness of Decentralized Multi-Agent Policy Evaluation under Byzantine Attacks}

\author{Hairi$^{\dag}$ \mbox{\hspace{0.4cm}} Minghong Fang$^{\P}$ \mbox{\hspace{0.4cm} Zifan Zhang$^{\ddag}$ \mbox{\hspace{0.4cm}}  Alvaro Velasquez$^{*}$ \mbox{\hspace{0.4cm}} Jia Liu$^{\S}$ \mbox{\hspace{0.4cm}}} 
\\ $^{\dag}$Dept. of Computer Science, University of Wisconsin-Whitewater
\\ $^{\P}$Dept. of Computer Science and Engineering, University of Louisville
\\ $^{\ddag}$Dept. of Computer Science, North Carolina State University
\\ $^{*}$Dept. of Computer Science, University of Colorado, Boulder
\\ $^{\S}$Dept. of Electrical and Computer Engineering, The Ohio State University
\thanks{
In Proceedings of the 22nd International Symposium on Modeling and Optimization in Mobile, Ad hoc, and Wireless Networks (WiOpt 2024). \\
Hairi and Minghong Fang are co-primary authors. 
Correspondence to: Hairi <hairif@uww.edu>, Minghong Fang <minghong.fang@louisville.edu>.
}
}

\maketitle

\begin{abstract}
In this paper, we study a fully-decentralized multi-agent policy evaluation problem, which is an important sub-problem in cooperative multi-agent reinforcement learning, in the presence of up to $f$ faulty agents. In particular, we focus on the so-called Byzantine faulty model with model poisoning setting. In general, policy evaluation is to evaluate the value function of any given policy. In cooperative multi-agent system, the system-wide rewards are usually modeled as the uniform average of rewards from all agents.  
We investigate the multi-agent policy evaluation problem in the presence of Byzantine agents, particularly in the setting of heterogeneous local rewards. Ideally, the goal of the agents is to evaluate the accumulated system-wide rewards, which are uniform average of rewards of the normal agents for a given policy. It means that all agents agree upon common values (the consensus part) and furthermore, the consensus values are the value functions (the convergence part).
However, we prove that this goal is not achievable.
Instead, we consider a relaxed version of the problem, where the goal of the agents is to evaluate accumulated system-wide reward, which is an appropriately weighted average reward of the normal agents.
We further prove that there is no correct algorithm that can guarantee that the total number of positive weights exceeds $|\mathcal{N}|-f $, where $|\mathcal{N}|$ is the number of normal agents. 
Towards the end, we propose a Byzantine-tolerant decentralized temporal difference algorithm that can guarantee asymptotic consensus under scalar function approximation. We then empirically test the effective of the proposed algorithm.
\end{abstract}


\section{Introduction}
Reinforcement learning (RL) \cite{SutBar_18} is a powerful paradigm in learning sequential decision-making. The success of RL both in theory \cite{AgaKakShaMah_21,JinZhuBub_18,SriYin_19,Sut_88,TsiVan_99,SutMcASin_99,XuWanLia_20} and practice \cite{SilHubJul_18,ShaShaSha_16,FoeFarAfo_18,Lin_22} has sparked the interest in the realm of multi-agent reinforcement learning (MARL) \cite{ZhaYanLiu_18,ZhaYanBas_21,JinLiuWan_21}. MARL \cite{Sha_53,ZhaYanLiu_18} is a multi-agent setting, a natural extension of single-agent RL, where agents interact within a common environment. The state dynamics and individual rewards are affected by both the global state and joint actions. Based on the system objective, there are in general two main categories of MARL problems, cooperative \cite{ZhaYanLiu_18} and competitive \cite{Sha_53} settings. Based on the assumption of the system infrastructure, there are also two categories, centralized setting and fully decentralized setting respectively. 
More specifically, in a fully-decentralized multi-agent setting, agents are only able to share information with each other through a communication network instead of a central server. In contrast, in a server-present centralized system, the server can collect and aggregate local information and disseminate appropriate information to agents (see an excellent survey \cite{ZhaYanBas_21} of MARL topics for further details). The focus of this paper is the cooperative and decentralized setting as in \cite{ZhaYanLiu_18}, where all agents work together to maximize a common goal. 

Similar to the single-agent RL setting, a complete MARL algorithm searches for a certain optimal policy $\pi^{*}$ that can maximize accumulated system-wide average reward, i.e.,
\begin{align}
\pi^{*}=\argmax_{\pi} \mathbb{E}\bigg[\sum_{t=0}^{\infty}\gamma^{t}\sum_{i=1}^{n}\frac{1}{n}r^{i}_{t+1}\bigg], \nonumber
\end{align}
where $n$ is the number of agents in the system, $\gamma$ is a discount factor with $0<\gamma<1$ and the expectation is subject to the usual caveats about appropriate expectations existing in steady-state. We note that in a cooperative multi-agent system, the system-wide reward is typically modeled as the uniform average of all agents. An important subproblem is to study the multi-agent policy evaluation for a given policy $\pi$, as this can be incorporated into the actor-critic framework as the critic step. The goal of all agents, in this subproblem, is to learn the value functions defined as:
\begin{align}
V(s)=\mathbb{E}\bigg[\sum_{t=0}^{\infty}\gamma^{t}\sum_{i=1}^{n} \frac{1}{n} r^{i}(s_t,a_t)|s_0=s\bigg], \nonumber
\end{align}
for all states $s\in\mathcal{S}$. This implies that 1) all agents need to reach consensus and 2) the consensus values are the value functions defined above. The multi-agent policy evaluation problem has been studied extensively in fault-free setting \cite{DoaMagRom_19,DoaMagRom_21,ZhaYanLiu_18,HaiLiuLu_22,CheZhoChe_21}.

We study a fully decentralized multi-agent policy evaluation problem in the presence of Byzantine agents. In addition, we consider a multi-agent system where up to $f>0$ agents are Byzantine.
Specifically, we explore the model poisoning faulty setting described in~\cite{CheSuXu_17,FanCaoJia_20,SuVai_16}, where Byzantine agents could send arbitrary or carefully crafted information to their neighboring agents.
%
%
In a fully decentralized system, it is typical for agents to share certain system parameters in order to achieve consensus as described above. However, Byzantine agents have the ability to modify these local parameters to arbitrary values, thereby disrupting the algorithm.
Furthermore, it is important to highlight that in a fully decentralized system, a Byzantine agent can transmit inconsistent information to its neighbors. This means that a Byzantine agent can send different values to different neighbors. This presents a significant challenge compared to the centralized server-based setting, where a Byzantine agent can only send a single piece of information to the server.
The existing literature in multi-agent reinforcement learning (MARL) lacks a comprehensive study on robust designs, particularly in heterogeneous settings that consider these challenges.

In this work, we investigate the multi-agent policy evaluation problem in the presence of Byzantine agents for any given policy $\pi$. Ideally, the goal of the agents is to evaluate the accumulated uniform average reward of the normal agents. Specifically, let $\mathcal{N}$ denotes the set of normal agents in the system, the decentralized multi-agent policy evaluation is to characterize the following value at any states $s$ for the given policy $\pi$:
\begin{align}
V(s)=\mathbb{E}\bigg[\sum_{t=0}^{\infty}\gamma^{t}\sum_{i\in\mathcal{N}} \frac{1}{|\mathcal{N}|} r^{i}(s_t,a_t)|s_0=s\bigg].
\label{eq: uniform_weight_value_function}
\end{align}

However, we will prove later in Theorem~\ref{thm: byzn_unif_ave} that evaluating Eq.~\eqref{eq: uniform_weight_value_function} {\em cannot} be achieved.
Thus, we consider a {\em relaxed} version of the multi-agent policy evaluation problem. In this relaxed problem, the goal of the agents is to evaluate accumulated weighted average reward, which can be written as:
\begin{align}
V(s)=\mathbb{E}\bigg[\sum_{t=0}^{\infty}\gamma^{t}\sum_{i\in\mathcal{N}} \alpha_i r^{i}(s_t,a_t)|s_0=s\bigg],
\label{eq: alpha_weight_value_function}
\end{align}
where $\alpha_i\ge 0$ for all $i\in\mathcal{N}$ and $\sum_{i\in\mathcal{N}}\alpha_i=1$.
We further prove that for the case $f>0$, there is {\em no} correct algorithm that can evaluate Eq.~\eqref{eq: alpha_weight_value_function} with $\sum_{i\in\mathcal{N}}\mathbf{1}\{\alpha_i>0\}> |\mathcal{N}|-f $, where $|\mathcal{N}|$ and $f$ are number of normal and the maximum number of Byzantine agents. In other words, achieving more than $|\mathcal{N}|-f$ positive weights in the relaxed problem is impossible in general.
In the end, we propose a Byzantine-tolerant decentralized temporal difference (BDTD) algorithm under linear scalar function approximation that can guarantee that all normal agents reach consensus.

The contributions of this paper are  threefolds:
\begin{list}{\labelitemi}{\leftmargin=1em \itemindent=-0.0em \itemsep=.1em}
\item First, we prove in Theorem \ref{thm: byzn_unif_ave} that evaluating the exact value functions defined by the uniform average reward of the agents in the presence of Byzantine is impossible. In other words, there is no correct algorithm that can achieve the value functions where system-wide rewards are modeled as the uniform average rewards of all normal agents in the presence of Byzantine agents. 
We further relax the problem to consider solving value function where system-wide rewards are modeled as appropriately weighted average rewards of the normal agents. 

\item Second, we further prove in Theorem \ref{thm: byzn_wei_ave} that there is no correct algorithm that can guarantee the number of positive weights exceeds $|\mathcal{N}|-f$ for the aforementioned relaxed problem.

\item Last but not least, we propose a decentralized multi-agent policy evaluation algorithm with linear scalar function approximation, so that all normal agents can reach consensus.
\end{list}


\section{Related work}

\subsection{Fault-free policy evaluation}
Policy evaluation, which aims to evaluate how good a given policy is, is an important sub-problem in designing a complete RL algorithm, which can be incorporated into the actor-critic framework as the critic step. 
Temporal difference (TD) learning \cite{Sut_88} is a simple yet effective learning algorithm first proposed in the single-agent setting to evaluate a given policy. 
The convergence theory in TD learning has been developed first in asymptotic regime \cite{TsiVan_99,TsiVan_02} and then in finite-time horizon \cite{BhaRusSin_18,SriYin_19,XuWanLia_20}.

The multi-agent policy evaluation, based on distributed TD learning, has been recently studied \cite{DoaMagRom_19,DoaMagRom_21,WuSheChe_21}. 
Various aspects of fully-decentralized MARL algorithms have been studied.
Notably, the sample and communication efficiencies of actor-critic algorithms have been investigated in \cite{HaiLiuLu_22,CheZhoChe_21,LiuWeiYin_22,HaiZhaLiu_24}.  

\subsection{Distributed Learning with Byzantine Agents}
Byzantine agents with local model poisoning attack is a common modeling for robust design of distributed algorithm design.
A large body of papers~\cite{CheSuXu_17,YinCheKan_18,FanCaoJia_20,fang2022aflguard,blanchard2017machine,he2022byzantine,xie2018generalized,karimireddy2020byzantine,karimireddy2021learning} in the literature have adopted it as a common failure model in federated learning problem, where a server is involved to facilitate the collaborative learning process within the supervised setting. In robust algorithm design, one feature that is different from fault-free counterpart is to design robust filtering mechanism.
For instance, 
in the Krum aggregation rule, as described by~\cite{blanchard2017machine}, the server receives local models from agents and selects one received local model that has the smallest distance to its subset of neighbors as the output.
In~\cite{cao2020fltrust}, a key system assumption is that the server holds a trusted dataset.
The server maintains a server model based on the current global model and its trusted dataset. Upon receiving one local model from any agent, the server considers this received local model as benign if it is positive related to the server model.
Recently work in \cite{CheZhaZha_23,FanMaDai_21} studied the effect of Byzantine agents in the so-called federated reinforcement learning (FRL) framework, where a central server is assumed to be present. 
However, we note that FRL and MARL differ significantly in that FRL is a multiple independent identical learner and the action from one agent does not affect the outcomes of other agents. In contrast, the global state transition and local rewards are dependent upon joint actions in MARL. In \cite{CheZhaZha_23}, the results are further extended the results to the offline setting.
The closest related work \cite{WuSheChe_21} studied the policy value evaluation in the presence of Byzantine agents for a given policy. However, the analysis implicitly assumes the setting of homogeneous rewards, i.e. the rewards for all agents are the same. In our work, we consider a more general heterogeneous reward setting. 
The offline competitive MARL has been studied in \cite{WuMcMZhu_22}, where the data poisoning fault model is considered. Specifically, the rewards in the offline data are adversarially changed so that the new Nash equilibrium learned from the poisoned data is significantly different from the Nash equilibrium learned from the original data.

There are a series of works \cite{SuVai_15a,SuVai_15c,SuVai_16} on decentralized optimization problems where the local objective functions are heterogeneous and convex. 
An important subproblem in both our work and work in decentralized optimization \cite{SuVai_15a,SuVai_15c,SuVai_16} is decentralized consensus, meaning all agents are required to agree with each other. Existing work in \cite{Vai_12,VaiTseLia_12} have focused on these fundamental problems and proposed $f$-trimmed-mean-based algorithms. A recent paper \cite{FanZhaHai_24} has investigated on the topic of Byzantine-robust decentralized federated learning.
\section{Byzantine Policy Evaluation in Multi-Agent Reinforcement Learning} \label{sec:formulation}
Throughout this paper, $\|\cdot\|$ denotes the $\ell_2$-norm for vectors and the $\ell_2$-induced norm for matrices. $|\cdot|$ denotes cardinality of a set/multi-set or the absolute value of a scalar.
$(\cdot)^{T}$ denotes the transpose for a matrix or a vector.

{\bf 1) System model:}
Consider a multi-agent system with $n$ agents, including up to $f$ agents to be Byzantine agents.
We denote the set of Byzantine agents as $\mathcal{F}$.
Note that the actual number of Byzantine agents in the system can be smaller than $f$. 
We consider the scenario that all $n$ agents are connected through a complete graph, where each edge serves as a communication channel that allows agents to send information to their neighbors.
Later, we will show that our impossibility results hold even for this most ideal setting.

\begin{definition}[Networked Multi-Agent MDP]
Let the communication network be a complete graph. A networked multi-agent MDP is defined by following tuple $(\mathcal{S},\{\mathcal{A}^{i}\}_{i=1}^{n},P,\{r^{i}\}_{i=1}^{n},\gamma)$, where $\mathcal{S}$ is the global state space observed by all agents, $\mathcal{A}^{i}$ is the action set for agent $i$,
$P:\mathcal{S}\times\mathcal{A}\times\mathcal{S}\to[0,1]$ is a global state transition kernel, $r^{i}: \mathcal{S}\times\mathcal{A}$ is the local reward function for agent $i$, and $\gamma\in(0,1)$ is the discount factor. 
Let $\mathcal{A}=\prod_{i\in\mathcal{N}}\mathcal{A}^i$ be the joint action set of all agents. \label{def: model}
\end{definition}

In this paper, we assume that the global state space $\mathcal{S}$ is finite.
We also assume that at any given time $t\ge0$, all agents can observe the current global state $s_t$. 
$r^{i}(s,a)$ is individual agent $i$'s reward given global state $s$ and joint action $a$.
For simplicity of the presentation, we assume that the rewards are deterministic. 
Even in this simple setting, we will show that our impossibility results hold, let alone for more general stochastic reward settings.
We consider policies that are stationary. 
In our MARL system, each agent chooses its action following its local policy $\pi^i$ that is conditioned on the current global state $s$, i.e., $\pi^{i}(a^{i}|s)$ is the probability for agent $i$ to choose an action $a^{i}\in\mathcal{A}^{i}$. 
Then, the joint policy $\pi: \mathcal{S}\times\mathcal{A}\to[0,1]$ can be written as $\pi(a|s)=\prod_{i\in\mathcal{N}}\pi^{i}(a^{i}|s)$. 
For any given policy $\pi$, the global value function for all $s\in\mathcal{S}$ is defined as follows: $V(s)= \mathbb{E}_{s\sim d_{\pi},a\sim \pi(\cdot|s)} [ \sum_{t=0}^{\infty}\frac{\gamma^{t}}{N}\sum_{i\in\mathcal{N}} r^{i}(s_t,a_t)
| s_0=s ]$,
where $d_{\pi}(\cdot)$ is the steady state distribution induced by $\pi$. The existence of such distribution is guaranteed by the Assumption \ref{ass: dis}.

\begin{definition}[Byzantine Networked Multi-Agent MDP]
A Byzantine networked multi-agent MDP is a networked multi-agent MDP as defined in Definition~\ref{def: model} with up to $f$ Byzantine agents, who may send arbitrary information when sharing to the neighboring agents. 
\label{def: byzantine_model}
\end{definition}

We note that in the modelling of the Byzantine agents, the agents still strictly follow the sampling policies and receives true data from the environment. However, the Byzantine behavior appears in the communication process with neighboring agents when sending value function information. One can see such modelling in Algorithm \ref{alg: Byzantine}.

{\bf 2) Technical assumptions:}
We now state the following assumptions for the decentralized multi-agent MDP described above. 
\begin{assumption}
For any policy $\pi$, the induced Markov chain $\{s_t\}_{t\ge0}$ is irreducible and aperiodic.
\label{ass: dis}
\end{assumption}

\begin{assumption}
The reward $r^{i}_{t+1}$ is uniformly bounded by a constant $r_{\max}>0$,  $\forall i\in [n]$ and $t\ge0$.
\label{ass: r_bou}
\end{assumption}

\begin{assumption}
Each agent $i$'s value function is parameterized by linear functions, i.e., $V(s;w)=\phi(s) w$, where $\phi(s)\in \mathbb{R}^{d}$ is a feature vector for state $s\in\mathcal{S}$. The feature matrix $\Phi\in \mathbb{R}^{|\mathcal{S}|\times d}$ is a full-rank matrix. The feature vectors $\phi(s)$ are uniformly bounded for any $s\in\mathcal{S}$. Without loss of generality, we assume that $\|\phi(s)\| \le 1$. 
\label{ass: fea}
\end{assumption}

\begin{assumption} \label{ass: Byzantine_upper_bound}
The total number of agents $n$ and the maximum number of Byzantine agents $f$ has the following inequality $n\ge 3f+1.$
\end{assumption}

Assumption~\ref{ass: dis} guarantees that there exists a stationary distribution $d_{\pi}(\cdot)$ over $\mathcal{S}$ for the Markov chain induced by the given policy $\pi$. Assumption~\ref{ass: r_bou} is common in the RL literature (see, e.g., \cite{ZhaYanLiu_18,XuWanLia_20,DoaMagRom_19}) and easy to be satisfied in many practical MDP models with finite state and action spaces.
Assumption~\ref{ass: fea} on features is standard and has been widely adopted in the literature, e.g., \cite{TsiVan_99,ZhaYanLiu_18,SriYin_19}, for linear function approximations. 
Assumption \ref{ass: Byzantine_upper_bound} is a standard assumption in decentralized Byzantine consensus problem as in \cite{Vai_12}.
\section{General results in Byzantine faulty multi-agent policy evaluation} \label{sec: gen_res_byz}

In this section, we start with the scope of the problems that we consider to facilitate the later discussions on our impossibility results. 
First, we introduce the Byzantine-free multi-agent policy evaluation problem \cite{ZhaYanLiu_18,DoaMagRom_19,HaiLiuLu_22,CheZhoChe_21}. 
Note that the stochastic convergence we are referring to in this paper is \emph{expected mean-squared convergence} as in Byzantine-free setting \cite{SriYin_19,DoaMagRom_19,HaiLiuLu_22,CheZhoChe_21}. 
{\bf 1) Byzantine-free multi-agent policy evaluation problem:}
\begin{problem}
In decentralized TD learning, if all agents function normally, is there a correct distributed TD learning algorithm converges to a TD fixed point that satisfies: $w^{*}=\mathbb{E}_{s\sim d_{\pi}(\cdot),a\sim \pi(\cdot|s)}[\phi(s)\sum_{i=1}^{n} \frac{1}{n}r^{i}(s,a)]$?
\label{pro: byzn_free}
\end{problem}

Note that the above convergence requires that each normal agent $i$ satisfy $\lim_{t\to\infty}\mathbb{E}[\|w^{i}_t-w^{*}\|^{2}]=0$.
This implies two important details. 
First, it signifies that all agents' parameters achieve consensus, meaning that they will all have the same values. 
Second, in addition to reaching a consensus, the agreed-upon value is $w^*$, which is referred to as the TD-fixed point.
We further note that from the perspective of the actor-critic framework in decentralized MARL, consensus on certain global information like value function is essential in computing local policy gradients \cite{ZhaYanLiu_18,HaiLiuLu_22,CheZhoChe_21}.

{\bf 2) Byzantine faulty multi-agent policy evaluation problems:}
With the presence of Byzantine agents, since we consider model poisoning Byzantine attack, it is clearly impossible to guarantee Byzantine agents to converge to the aforementioned Byzantine-free TD-fixed point $w^{*}$, defined in Problem~\ref{pro: byzn_free}. 
A natural goal is to consider if there exist correct algorithms such that the parameters converge to the fixed point corresponding to normal agents, 
which is formally stated as follows.
\begin{problem}
When $f>0$, is there a correct TD learning algorithm that allows the agents to converge to
$$w^{*}_{\mathcal{N}}=\mathbb{E}_{s\sim d_{\pi}(\cdot),a\sim \pi(\cdot|s)}[\phi(s)\sum_{i\in\mathcal{N}} \frac{1}{|\mathcal{N}|}r^{i}(s,a)],$$
where $\mathcal{N}$ denotes the set of normal agents?
\label{pro: byzn_unif_ave}
\end{problem}

The fixed point $w^{*}_{\mathcal{N}}$ proposed corresponds to modelling the system rewards as the uniform average of all normal agents. However, as we will prove in Theorem~\ref{thm: byzn_unif_ave}, it is impossible to reach the TD-fixed point defined in Problem~\ref{pro: byzn_unif_ave}. 
Thus, we further relax the problem to consider a TD fixed point that is an appropriately weighted average of all normal agents.

\begin{problem}
When $f>0$, is there a correct TD learning algorithm that allows the agents to converge to $$w^{*}_{\alpha}=\mathbb{E}_{s\sim d_{\pi}(\cdot),a\sim \pi(\cdot|s)}[\phi(s)\sum_{i\in\mathcal{N}} \alpha_i r^{i}(s,a)],$$
where the weights $\alpha_i$ satisfies:
$\sum_{i\in\mathcal{N}} \alpha_i =1$, $\alpha_i \ge 0,  \forall i\in\mathcal{N}$?
\label{pro: byzn_wei_ave}
\end{problem}

The fixed point $w^{*}_{\alpha}$ proposed corresponds to modelling the system rewards as a non-uniform weighted average of all normal agents.
In Theorem~\ref{thm: byzn_wei_ave}, we will answer this question formally. 
In general, there is no correct algorithm that can guarantee the number of positive weights exceeds $|\mathcal{N}|-f$. 
In other words, in some multi-agent policy evaluation problems, achieving $|\mathcal{N}|-f$ number of positive weights is the best an algorithm can do. 
Toward this end, we introduce a $(\nu,\xi)$-admissible problem.
\begin{problem}
($(\nu,\xi)$-admissible problem) When $f>0$, for given pair of $\nu\in\mathbb{N}^{+}$ and $\xi>0$, is there a correct TD learning algorithm that allows the agents to converge to 
\begin{align}
w^{*}_{\nu,\xi}=\mathbb{E}_{s\sim d_{\pi}(\cdot),a\sim \pi(\cdot|s)}[\phi(s)\sum_{i\in\mathcal{N}} \alpha_i r^{i}(s,a)].
\label{eq: byzantine-alpha_TD}
\end{align}
where the weights $\alpha_i$ satisfies $\sum_{i\in\mathcal{N}} \alpha_i =1$, $\alpha_i \ge 0$, $\forall i\in\mathcal{N}$, $\sum_{i\in\mathcal{N}} \mathbf{1}(\alpha_i\ge \xi) \ge \nu$?
\label{pro: byzn_adm}
\end{problem}
Problem~\ref{pro: byzn_adm} is to learn the value functions with at least $\nu$ positive weights, which are bounded away from zero by at least $\xi$. 
It is easy to see that when $\xi=0$ and $\nu=1$, Problem~\ref{pro: byzn_adm} reduces to Problem~\ref{pro: byzn_wei_ave}.

{\bf 3) Main theoretical results:}
The following theorems state that, in the presence of Byzantine agents, no algorithm ensures that the normal agents' parameters converge to a fixed point in Problem~\ref{pro: byzn_unif_ave}. 
\begin{theorem}
When $f>0$, Problem \ref{pro: byzn_unif_ave} is not solvable.
\label{thm: byzn_unif_ave}
\end{theorem}

\begin{theorem}
For any $\xi>0$, Problem \ref{pro: byzn_adm} is not solvable for any $\nu>|\mathcal{N}|-f$.
\label{thm: byzn_wei_ave}
\end{theorem}

Theorem~\ref{thm: byzn_wei_ave} says that a $(|\mathcal{N}|-f,\xi)$ admissible solution is the best one can achieve for some $\xi>0$. 

We remark that even though the proofs for above two theorems are inspired by \cite{SuVai_16}, there are two significant differences in the proofs and implications. First, our proof is convergence for stochastic terms whereas in \cite{SuVai_16}, the proof is for deterministic terms. Secondly, the impossibility results hold for general multi-agent policy evaluation problem, including tabular case and linear approximations, whereas in \cite{SuVai_16}, the impossibility result is just for scalar case. 

We also remark that the impossibility results holds for general graph, not just limited to complete graph, in decentralized multi-agent settings as well where the proof will be the same. The reason that we assumed a complete graph in the beginning is to design the algorithm in Section \ref{sec: BDTD}.
\section{Byzantine-tolerant Decentralized Temporal Difference Learning}
\label{sec: BDTD}
In this section, we provide a Byzantine-tolerant decentralized TD (BDTD) learning algorithm for normal agents to solve MARL policy evaluation in the sense of Theorem~\ref{thm: byzn_wei_ave}. In order to derive such an algorithm, we further assume in Assumption \ref{ass: fea}, the dimension $d=1$, i.e. the features are reduced to scalar features.

{\bf 1) Behavior of Byzantine agents:}
The behaviors of Byzantine agents are described in Algorithm~\ref{alg: Byzantine}. 
The parameters sent by the Byzantine agents can be arbitrary (denoted as $*$). 
We note that Byzantine agents can only poison the local models of their own, which are the information to be exchanged with their neighbors. 
This is referred to as local model poisoning \cite{FanCaoJia_20}. 
We do not consider the data poisoning models, where Byzantine agents may change the data which may include local policies and local actions (global state as a result). 

On the other hand, in a decentralized multi-agent setting, a Byzantine agent can send inconsistent parameters to its neighbors, which means that a Byzantine agent can send one parameter to one neighbor and a distinct parameter to another neighboring agent. 
There is a more restricted Byzantine model called Byzantine broadcast model \cite{SuVai_15a}, where a Byzantine agent sends the same parameter to neighboring agents. 
Here, in our work, we focus on the more general setting where Byzantine agents may send inconsistent parameters.

\begin{algorithm}[t]
  \SetKwInOut{Input}{Input}
  \SetKwInOut{Output}{Output}
  \Input{initial state $s_0$, given policy $\pi$, state features $\phi$, step-size $\eta_k$, initial parameters $\{w^{i}_0\}_{i\in\mathcal{V}}$.}
  \BlankLine
    \For{$k=0,1,\cdots$}{
  	 \For{ all $i\in\mathcal{F}$}{
		  Execute action $a^{i}_{k}\sim\pi^{i}(\cdot|s_{k})$; \\
		  Observe the state $s_{k+1}$ and reward $r^{i}_{k+1}$;
		  }
          Send $*$ to neighbors\footnotemark and receive values from neighbors;
  }
  \caption{\textbf{Byzantine Agent's Behavior}.}
  \label{alg: Byzantine}
\end{algorithm}
\footnotetext[1]{The arbitrary value * can be different to neighbors.}

{\bf 2) $f$-Trimmed mean subroutine:}
We will define $f$-trimmed mean, which is a subroutine we used for parameters.
\begin{definition}[$f$-Trimmed Mean~\cite{YinCheKan_18}]
For any multi-set\footnote{A multi-set allows the elements in it to be the same.} $\{x^{1},\cdots,x^{n}\}$, where $x^{i}\in\mathbb{R}$ for all $i$, sort the $n$ values in ascending order (break the tie uniformly random), then remove the largest $f$ and smallest $f$, respectively. For the remaining $n-2f$ values, return the average value.
\end{definition}

{\bf 3) Policy evaluation for normal agents:}
Algorithm \ref{alg: critic} describes the decentralized multi-agent policy evaluation algorithm for normal agents. For any given policy $\pi$, the algorithm learns the value function parameters using decentralized TD learning.\footnote{For simplicity, we used TD(0) instead of TD($\lambda$).
The extension to TD$(\lambda)$ where $\lambda\in(0,1]$ is straightforward.} We note that for normal agents $i\in\mathcal{N}$, it is only required to know its own local policy $\pi^{i}$.

In Line-2, if a value is not received from neighbors, set the value to be some default value; If a value is beyond the projection ball, $n\leftarrow n-1$, $f\leftarrow f-1$ and remove the corresponding agent.
In Line~9, we have used projected TD learning, a variant of TD learning introduced in \cite{BhaRusSin_18}. 
A choice for such radius $R$ in our scalar case is $R=\frac{2r_{\max}}{\phi_{\min}(1-\gamma)^{3/2}}$, where $\phi_{\min}:=\min_{s\in\mathcal{S}}|\phi(s)|$ (see \cite[Lemma~7]{BhaRusSin_18} for vector case). 
This projection step is mainly for theoretical analysis for bounding TD error terms.
In practice, such a projection step may be dropped.
The step sizes $\eta_t$ used in Line~8 of Algorithm~\ref{alg: critic} are diminishing. 
The step sizes are known to all agents as priori and satisfy the standard conditions: $\sum_{t=1}^{\infty} \eta_t=\infty$ and $\sum_{t=1}^{\infty} \eta^{2}_t<\infty$. 
A typical step size choice is $\eta_t=\frac{1}{t}$ for $t\ge 0$.

\begin{algorithm}[t!]
  \SetKwInOut{Input}{Input}
  \SetKwInOut{Output}{Output}
  \Input{initial state $s_0$, given policy $\pi$, state features $\phi$, step-size $\eta_k$}
  \BlankLine

  \For{$k=0,1,\cdots$}{
    Send parameter $w^{i}_k$ to neighbors and receive values from neighbors;\\ 
    Consensus Update: $\tilde{w}^{i}_{k}\leftarrow \text{$f$-Trimmed Mean}$;\\
  	\For{all $i\in\mathcal{N}$}{ 
        Execute action $a^{i}_{k}\sim\pi^{i}(\cdot|s_{k})$; \\
		Observe the state $s_{k+1}$ and reward $r^{i}_{k+1}$; \\
		Update $\delta^{i}_{k}\leftarrow r^{i}_{k+1}+\gamma\phi^{T}(s_{k+1})w^{i}_k-\phi^{T}(s_{k})w^{i}_k$;
		}
		\bf{Projected TD Step:} $w^{i}_{k+1}\leftarrow \Pi_{2,R}(\tilde{w}^{i}_k+\eta_k\delta^{i}_{k}\cdot\phi(s_{k}))$; 
  }
  \caption{\textbf{Byzantine-Tolerant Decentralized TD (BDTD) Learning for Normal Agents}.}
  \label{alg: critic}
\end{algorithm}

{\bf 4) Main theoretical results for Algorithm~\ref{alg: critic}}
Let $\bar{w}_t=\frac{1}{|\mathcal{N}|}\sum_{i\in\mathcal{N}}w^{i}_t$, i.e. the average of the parameters of normal agents at iteration $t\ge 0$. Then, we have the following consensus result that states parameters of all normal agents will converge to the average asymptotically.
\begin{theorem}
The sequences generated in normal agents by Algorithm \ref{alg: critic} will achieve consensus, i.e. for any $i\in\mathcal{N}$, we have $
\lim_{t\to\infty} |w^{i}_t-\bar{w}_t|=0.$
\label{thm: consensus}
\end{theorem}

Theorem~\ref{thm: consensus} ensures that even starting with different initial parameters, in the heterogeneous reward setting and, more importantly, inconsistent Byzantine faulty model, the parameters among normal agents will eventually reach an agreement. 
However, the average parameter $\bar{w}_t$ itself may not have a limit depending on the heterogeneity of the problem.

\section{Evaluation}

\subsection{Experimental Setup}
\label{sec:setup}

\begin{figure*}[!t]
\centering
\subfloat[Gaussian attack]{\includegraphics[width=0.28 \textwidth]{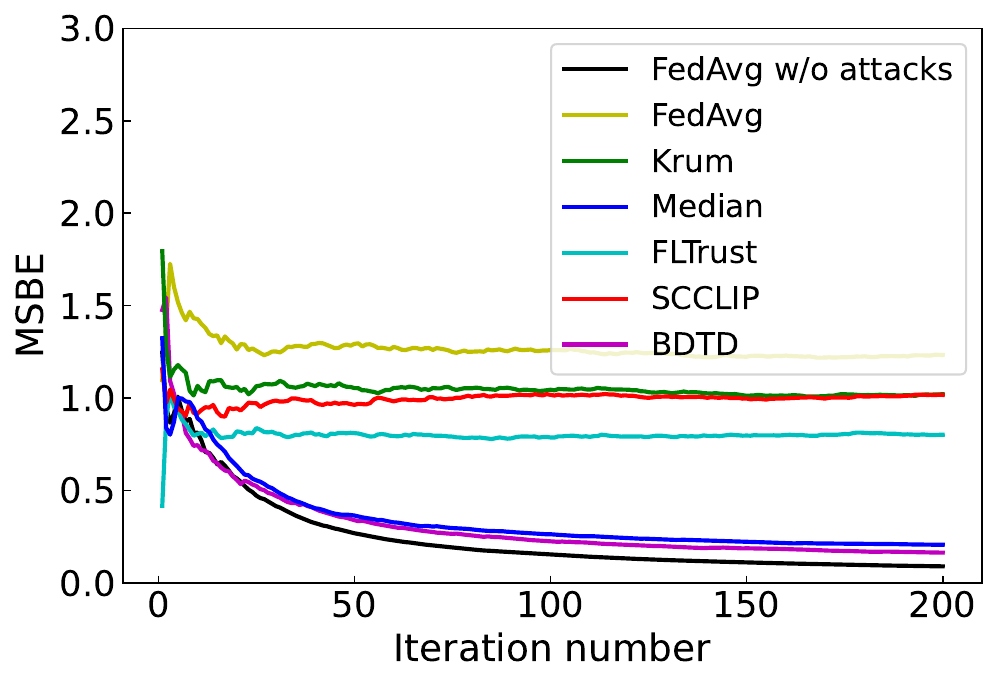}}
 \subfloat[Krum attack]{\includegraphics[width=0.28 \textwidth]{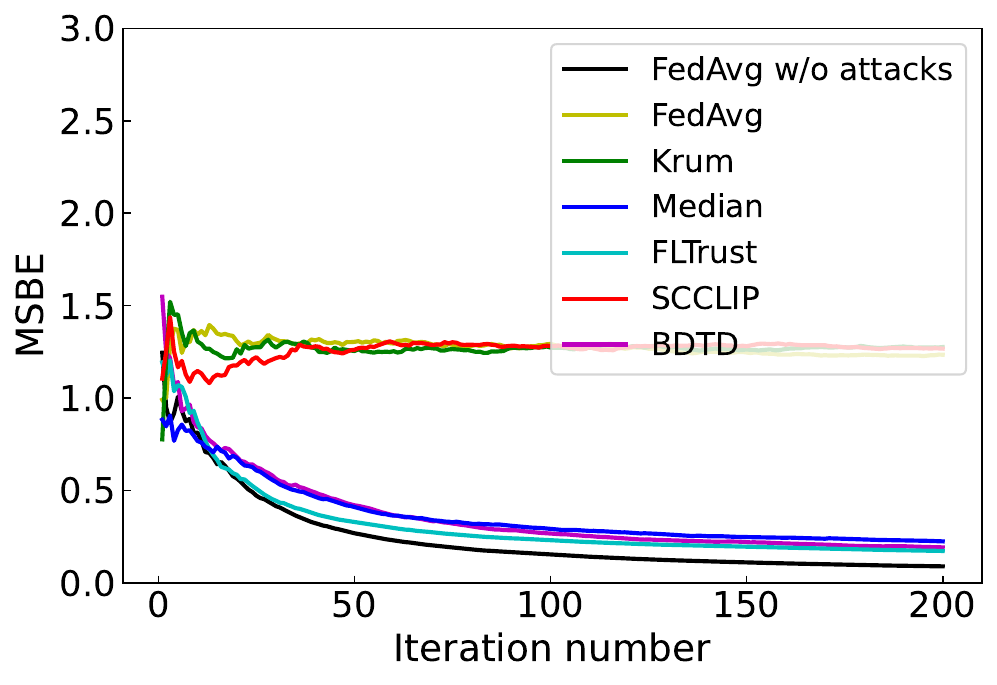}}
  \subfloat[Trim attack]{\includegraphics[width=0.28 \textwidth]{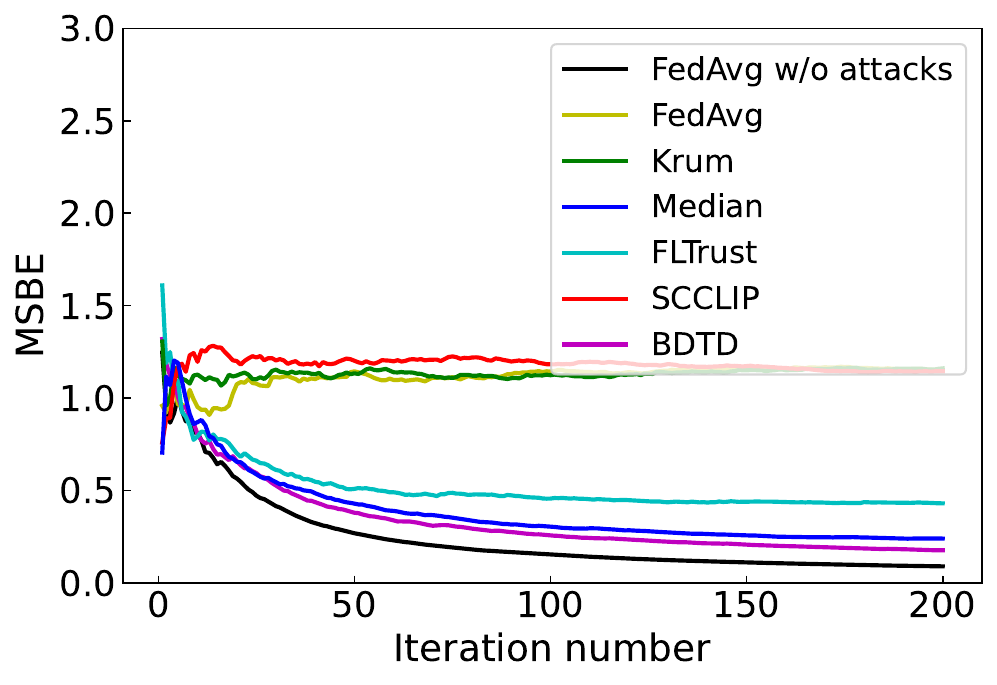}}
    \caption{Mean squared Bellman error (MSBE) of different methods under different attacks.}
  \label{result_msbe}
\vspace{-0.5mm}
\end{figure*}

\begin{figure*}[!t]
\centering
\subfloat[Gaussian attack]{\includegraphics[width=0.28 \textwidth]{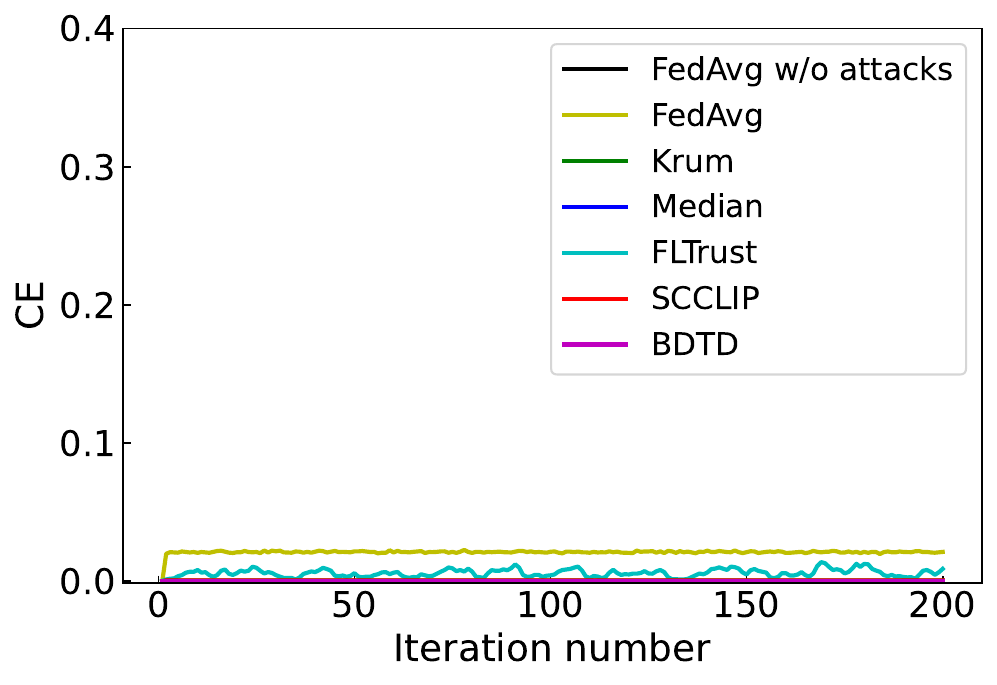}}
 \subfloat[Krum attack]{\includegraphics[width=0.28 \textwidth]{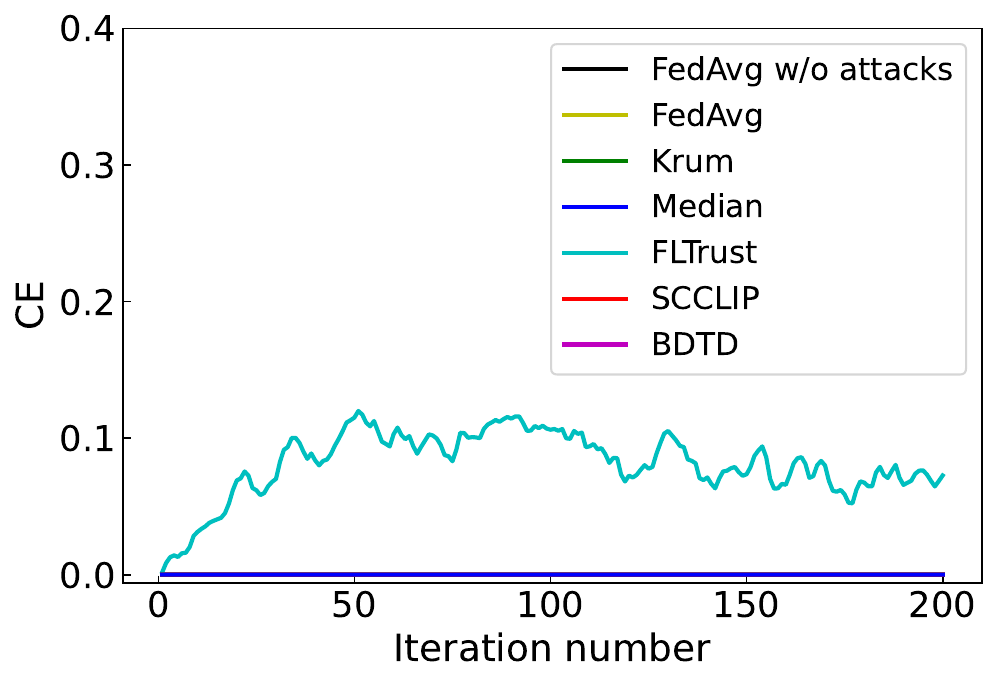}}
  \subfloat[Trim attack]{\includegraphics[width=0.28 \textwidth]{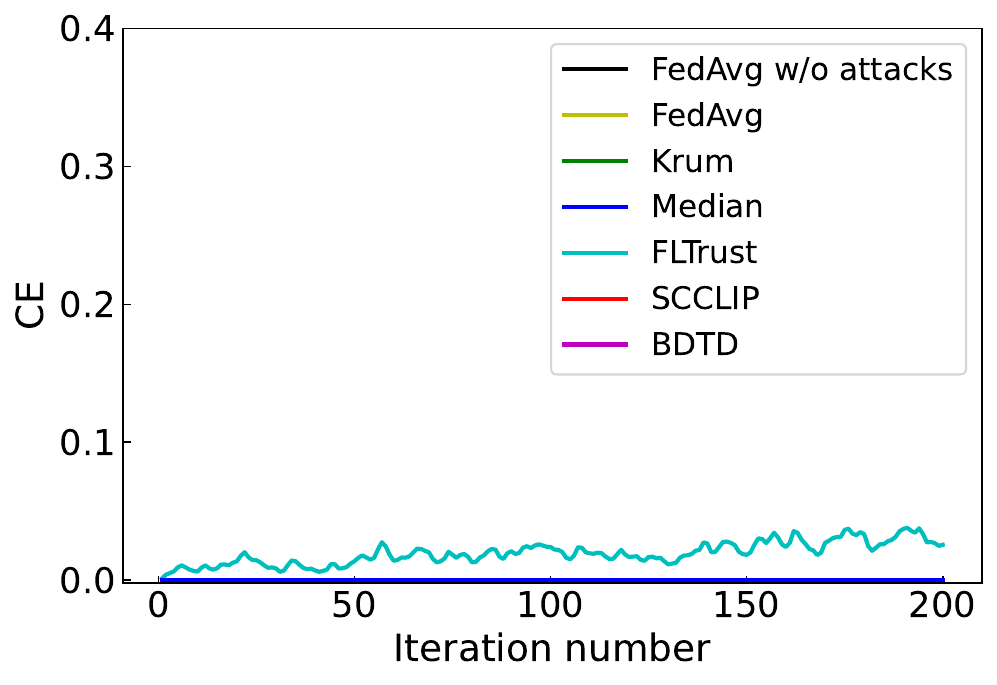}}
    \caption{Consensus error (CE) of different methods under different attacks.}
  \label{result_mce}
\vspace{-0.9mm}
\end{figure*}

{\bf 1) Parameter Settings:}
We consider a cooperative navigation task known as Simple Spread, derived from the Multi-Particle Environment (MPE)~\cite{lowe2017multi}. 
The task involves 10 agents aiming to collectively cover all landmarks. There are two malicious agents among them.
The agents receive rewards based on the proximity between the closest agent and each landmark. 
Collisions between agents result in negative rewards.
Each agent selects actions from the action space $\mathcal{A}=$\{no action, move left, move right, move down, move up\} using a uniformly random policy.
The objective is to train all agents to identify and cover their respective landmarks while avoiding collisions. The malicious agents, on the other hand, attempt to deceive the other agents by providing arbitrary information.
The feature dimension is 40, encompassing the agents' self-positions, relative positions of landmarks, and relative positions of other agents.
The step-size is set to 0.1.
We run our experiments on Intel(R) Core(TM) i9-12900K CPU.
We repeat each experiment 10 times, and report the average results.
Since the variances of results are small, we omit them here.

{\bf 2) Compared Methods:} 
We compare our \alg algorithm with the following aggregation baselines. 

\begin{list}{\labelitemi}{\leftmargin=1em \itemindent=-0.0em \itemsep=.1em}

\item {\em FedAvg~\cite{mcmahan2017communication}:} Every agent, upon receiving parameters from its neighboring agents, calculates the weighted mean of the received parameters.

\item {\em Krum~\cite{blanchard2017machine}:}
In the Krum aggregation rule, each agent produces a single parameter that minimizes the sum of distances to its subset of neighbors, and the size of the subset is $n-f$, where $n$ is the total number of agents and $f$ is the maximum number of Byzantine agents.

\item {\em Coordinate-wise median (Median) ~\cite{YinCheKan_18}:} In every dimension, each agent calculates the coordinate-wise median of all the parameters it receives.

\item {\em FLTrust~\cite{cao2020fltrust}:}
When an agent receives a parameter from its neighboring agent, it first calculates the cosine similarity between its own parameter and the received parameter. If the cosine similarity is positive, the agent then normalizes the received parameter to have the same magnitude as its own parameter. After that, the agent computes the weighted average of all the normalized parameters sent by its neighbors.

\item {\em SCCLIP~\cite{he2022byzantine}:}
The SCCLIP method mitigates the influence of Byzantine agents through the use of the clip operation. In this approach, when an agent receives parameters from its neighboring agents, it employs its own parameter as the reference point to limit or clip the received parameters.
\end{list}

{\bf 3) Poisoning Attacks:} 
We consider the following poisoning attack schemes in our experiments.
\begin{list}{\labelitemi}{\leftmargin=1em \itemindent=-0.0em \itemsep=.1em}
\item {\em Gaussian attack~\cite{blanchard2017machine}:}
In a Gaussian attack, each Byzantine agent samples a vector from a Gaussian distribution with a mean of zero and a standard deviation of one, then sends it to its neighboring agent.
\item {\em Krum attack~\cite{FanCaoJia_20}:}
In the Krum attack, Byzantine agents manipulate their parameters to degrade the Krum method's performance.
\item {\em Trim attack~\cite{FanCaoJia_20}:}
In the Trim attack, the attacker carefully manipulates the parameters of Byzantine agents in a way that causes a significant deviation between the aggregated parameter before and after the attack.
\end{list}

{\bf 4) Evaluation Metrics:}
We consider the following two evaluation metrics: i) mean squared Bellman error (MSBE) and ii) consensus error (CE).
Given parameters $\{w^{i}_k\}_{i\in\mathcal{N}}$ and samples $(s_k,s_{k+1})$, the empirical squared Bellman error (SBE) of the $\kappa$-th sample is defined as
$\text{SBE} (\left\{w_k^i\right\}_{i\in\mathcal{N}}, s_{\kappa},s_{\kappa + 1} )
:=\frac{1}{|\mathcal{N}|} \sum_{i \in \mathcal{N}}\left(\bar{r}_{\kappa}+\gamma\phi(s_{\kappa+1})^{T}w^{i}_{\kappa}-\phi(s_{\kappa})^{T}w^{i}_{\kappa}\right)^2$, where $\bar{r}_{\kappa}=\frac{1}{\mathcal{N}}\sum_{i\in\mathcal{N}} r^{i}_\kappa$.
Then, MSBE up to the $k$-th sample is defined as the average of SBEs over the history, which is computed as
$\text{MSBE}:=\frac{1}{k} \sum_{\kappa=1}^k \text{SBE} (\left\{w^i_\kappa\right\}_{i\in\mathcal{N}}, s_\kappa,s_{\kappa+1} )$.
The consensus error is computed as $\texttt{CE} = \frac{1}{|\mathcal{N}|}\sum_{i\in\mathcal{N}}\|w^{i}_k-\bar{w}_k\|^2$.
The smaller the MSBE and CE, the better the defense.

\vspace{-0.9mm}
\subsection{Experimental Results}
Figures~\ref{result_msbe} and~\ref{result_mce} show the MSBE and CE of different methods under different attacks.
``FedAvg w/o attacks'' means that there are no Byzantine agents in the system.
We observe from Figures~\ref{result_msbe} and~\ref{result_mce} that our proposed \alg overall achieves the best performance across various attack scenarios.
Even under the strong Trim attack, our proposed \alg's MSBE is comparable to that of FedAvg without any attacks. 
In contrast, existing Byzantine-robust aggregation rules, e.g., Krum and SCCLIP, are susceptible to poisoning attacks.
For instance, FLTrust is vulnerable to both the Gaussian and Krum attacks.
Under the Gaussian attack, the final MSBE of FLTrust is 0.801.
Similarly, under the Krum attack, although MSBE of FLTrust is low, CE is large, 
indicating a lack of consensus among the normal agents when using the FLTrust aggregation rule.
The Krum aggregation rule is susceptible to all three considered attacks.
Specifically, under three poisoning attacks, the CE of Krum remains small, but the MSBE becomes large.
This suggests that when normal agents employ the Krum aggregation rule, they tend to reach a poor consensus.

\section{Proofs for Theorems in Section \ref{sec: gen_res_byz}}
Let $[n]:=\{1,\cdots,n\}$, i.e. the set of all agents.

\subsection{Proof of Theorem \ref{thm: byzn_unif_ave}}
\begin{proof}
Assume that $f>0$. Inspired by \cite{SuVai_16}, the theorem is proved by contradiction.

Suppose that there exists a correct algorithm $\mathcal{A}$ that solves Problem \ref{pro: byzn_unif_ave}. Define the rewards of the $n$ agents as follows for all state-action pair $(s,a)\in \mathcal{S}\times \mathcal{A}$ to be
$r^{i}(s,a)=i$ for all $i\in [N]$.

Consider the following two executions that in the first one, agent $1$ is the Byzantine agent and the rest are normal agents whereas in the second one, agent $n$ is the Byzantine agent and the rest are normal agents. In both executions, Byzantine agent behaves correctly as the correct algorithm, this is reasonable as Byzantine agents can behave arbitrarily. As a result, for execution 1, algorithm $A$ outputs the result $w^{i,1}_t$ for each agent $i$ and given $t$ such that, we have 
\begin{align}
w^{i,1}_t&\xrightarrow{L^2}\frac{1}{n-1}\sum_{i\in\{2,\cdots,n\}} \mathbb{E}_{s\sim d_{\pi}(\cdot),a\sim \pi(\cdot|s)}[i\phi(s)] \nonumber \\
&=\frac{\sum_{i\in\{2,\cdots,n\}}i}{n-1} \mathbb{E}_{s\sim d_{\pi}(\cdot)}[\phi(s)] \nonumber \\
&=\frac{n(n+1)-2}{2(n-1)}\mathbb{E}_{s\sim d_{\pi}(\cdot)}[\phi(s)]\triangleq w^{*,1} \label{eq: execution_1}
\end{align}
where $L^{2}$ denote expected mean-square convergence. More specifically,
$$\lim_{t\to\infty}\mathbb{E}\|w^{i,1}_t-w^{*,1}\|^{2}=0.$$
Similarly, for execution 2, we have 
\begin{align}
w^{i,2}_t&\xrightarrow{L^2}\frac{1}{n-1}\sum_{i\in\{1,\cdots,n-1\}} \mathbb{E}_{s\sim d_{\pi}(\cdot),a\sim \pi(\cdot|s)}[i\phi(s)] \nonumber \\
&=\frac{\sum_{i\in\{1,\cdots,n-1\}}i}{n-1} \mathbb{E}_{s\sim d_{\pi}(\cdot)}[\phi(s)] \nonumber \\
&=\frac{n}{2}\mathbb{E}_{s\sim d_{\pi}(\cdot)}[\phi(s)]\triangleq w^{*,2}. \label{eq: execution_2}
\end{align}
Note that $$w^{*,1}-w^{*,2}=\mathbb{E}_{s\sim d_{\pi}(\cdot)}[\phi(s)].$$
By the assumption of linear independence of feature vectors $\phi(\cdot)$, by Assumption \ref{ass: fea}, and the fact that $d_{\pi}(\cdot)$ is a distribution, we know that there exists an entry in vector $\mathbb{E}_{s\sim d_{\pi}(\cdot)}[\phi(s)]$ is on-zero. As a result, $w^{*,1}\neq w^{*,2}$.

However, for any agent $i\in\{2,\cdots,n-1\}$ perspective, they can't distinguish the above 2 executions, as a result, they must output the same results for both executions. However, this contradicts with the assumption that both executions would converge to distinct fixed points shown in \eqref{eq: execution_1} and \eqref{eq: execution_2} respectively. Therefore, there's no correct algorithm exists for \emph{\textbf{Problem}} \ref{pro: byzn_unif_ave} and the proof is complete. 
\end{proof}

\subsection{Proof of Theorem \ref{thm: byzn_wei_ave}}
\begin{proof}
Recall that we assume $n>3f+1$ and denote the actual number of Byzantine agents in the system as $q$, i.e. $q=|\mathcal{F}|$. Let the rewards for any state-action pair $(s,a)\in\mathcal{S}\times\mathcal{A}$ and agent $i\in[n]$ to be 
\begin{align}
r^{i}(s,a)&=i, \quad \text{for } 1\le i\le f \text{ and } n-q+1\le i\le n \nonumber \\
r^{i}(s,a)&=f+1, \quad \text{for } f+1\le i\le n-q. \nonumber 
\end{align}
For any correct algorithm, consider the following two cases, where in both cases, Byzantine agents follow the correct algorithm.
\begin{itemize}
\item Case 1: In this case, agents $n-q+1\le i\le n$ are Byzantine agents. The output of the correct algorithm converges to
\begin{align}
w^{*}_{\alpha}\in [1,f+1]\mathbb{E}_{s\sim d_{\pi}(\cdot),a\sim \pi(\cdot|s)}[\phi(s)]. \nonumber 
\end{align}
\item Case 2:  In this case, agents $1\le i\le f$ are Byzantine agents. The output of the correct algorithm converges to
\begin{align}
w^{*}_{\alpha}\in [f+1,n]\mathbb{E}_{s\sim d_{\pi}(\cdot),a\sim \pi(\cdot|s)}[\phi(s)]. \nonumber 
\end{align}
\end{itemize}
As for any normal agent $i\in \{f+1,\cdots,n-q\}$ can't distinguish the above two cases, they must converge to an identical value in both cases. So, $w^{*}_{\alpha}$ must be $ (f+1)\mathbb{E}_{s\sim d_{\pi}(\cdot),a\sim (\cdot|s)}[\phi(s)]$. In other words, 
\begin{align}
w^{*}_{\alpha}&=(f+1)\mathbb{E}_{s\sim d_{\pi}(\cdot),a\sim (\cdot|s)}[\phi(s)] \nonumber \\
&=\sum_{i=1}^{n-q}\alpha_i r^{i} \mathbb{E}_{s\sim d_{\pi}(\cdot),a\sim \pi(\cdot|s)}[\phi(s)], \nonumber 
\end{align}
where the second equality is due to the definition of Case 1, which is equivalent to
\begin{align}
w^{*}_{\alpha}=\sum_{i=1}^{n-q}\alpha_i r^{i}=f+1. \label{eq: case_1}
\end{align}
The above equivalency is again because feature vectors are linearly independent.
By the reward setting given above, we further have, for \eqref{eq: case_1}, 
\begin{align}
\sum_{i=1}^{f}\alpha_i i+(f+1)\sum_{i=f+1}^{n-q}\alpha_i=f+1, \nonumber
\end{align}
which is equivalent to 
\begin{align}
\sum_{i=1}^{f}\alpha_i i=(f+1)(1-\sum_{i=f+1}^{n-q}\alpha_i)=(f+1)\sum_{i=1}^{f}\alpha_i 
\end{align}
which is only possible when $\alpha_i=0$ for all $1\le i\le f$. As a result, there could be at most $|\mathcal{N}|-f$ can be positive in Case 1 regardless of $\xi$. And $\nu$ can't be larger than  $|\mathcal{N}|-f$ and the proof is complete.
\end{proof}

\section{Conclusion}

In this paper, we studied fully decentralized multi-agent policy evaluation problem in the presence of Byzantine agents.
We first established the impossibility of designing a correct algorithm that obtains value functions where the system-wide rewards are represented as the uniform average rewards of all normal agents.
We then proceeded to relax the problem by considering the situation where the system-wide rewards are represented as appropriately weighted average rewards of the normal agents.
Subsequently, we demonstrated that there is no correct algorithm capable of ensuring that the number of positive weights surpasses $|\mathcal{N}|-f$ for the aforementioned relaxed problem.
Lastly, we proposed a decentralized multi-agent policy evaluation algorithm, which guarantees consensus among all normal agents.

\section*{Acknowledgments}
This work has been supported in part by NSF grants CAREER CNS-2110259, CNS-2112471, IIS-2324052, ECCS-2331104, a DARPA YFA Award D24AP00265, and an ONR grant N00014-24-1-2729.

\bibliographystyle{plain}
\bibliography{refs}

\appendix

\section{Proofs for Theorems in Section \ref{sec: gen_res_byz}}
\subsection{Consensus Derivation for Theorem \ref{thm: consensus}}
We note that the consensus update requires a new construction of a consensus matrix that is solely based on normal agents. That is because the dynamics of Byzantine agents are hard to predict and analyze even if it's possible. As a result, using Byzantine agents' parameters to characterize a normal agent's dynamic is not an approach to go. Fortunately, for the class of trim-mean based algorithms, it is possible to characterize a normal agent's behavior solely based on the normal neighbors of the agent's.

\begin{definition}[Reduced graph \cite{Vai_12}]
For a given graph $\mathcal{G}(\mathcal{V},\mathcal{E})$ and $\mathcal{F}\subset \mathcal{V}$, a graph $\mathcal{G}(\mathcal{N},\mathcal{E}_\mathcal{N})$ is said to be a reduced graph, if: (i) $\mathcal{N}=\mathcal{V}-\mathcal{F}$, and (ii) $\mathcal{E}_\mathcal{N}$ is obtained by first removing from $\mathcal{E}$ all the links incident on the nodes in $\mathcal{F}$, and then removing up to $f$ other incoming links at each node in $i\in \mathcal{N}$. 
\label{def: red_gra}
\end{definition}

Note that for complete graph that we consider in this work and a given $\mathcal{F}$, multiple reduced graph $\mathcal{G}$ may exist. Let $\mathcal{R}_F$ be the set of all possible reduced graphs $\mathcal{G}(\mathcal{N},\mathcal{E}_{\mathcal{N}})$. And further let $\tau=|\mathcal{R}_{F}|$, which is a finite number.

Since the underlying communication network is a complete graph and $n\ge 3f+1$, by \cite{Vai_12,SuVai_15c}, we know that the parameter update can be expressed with a transition matrix. 
In the following discussion, for any $\mathcal{H}\in R_{\mathcal{F}}$, we use $H$ to denote the connectivity matrix. That is, if there's a directed link from node $i$ to $j$, then $H(i,j)=1$ and otherwise $H(i,j)=0$. For diagonal elements, we set $H(i,i)=1$ for all $i\in\mathcal{N}$.
Moreover there exists a constant $\beta>0$ such that the following lemma holds.
\begin{lemma}[From \cite{Vai_12}]
 For any $t\ge 0$, there exists a graph $\mathcal{H}[t]\in R_{\mathcal{F}}$ such that $\beta H[t]\le A_t$.   
 \label{lem: A_low_bou}
\end{lemma}

For the detailed expression of $\beta$, see \cite{Vai_12}. A lower bound for $\beta$ is $\frac{1}{2f(f+1)}$, i.e. $\frac{1}{2f(f+1)}\le\beta$.
Recall the definition of reduced graph in Definition \ref{def: red_gra} and that $\tau$ is the cardinality of the set of reduced graph.

\begin{lemma}
In the product below of $H(t)$ matrices for consecutive $\tau |\mathcal{N}|$ iterations, at least one column is non-zero.
$$\prod_{t=k}^{k+\tau |\mathcal{N}|-1}H[t].$$
\label{lem: H_all_one}
\end{lemma}
\begin{proof}
Since the above product contains $\tau |\mathcal{N}|$ matrices, at least one matrix, say $H^{*}$, appeared at least $N$ times. In addition, since all diagonal elements of $H$ matrices are 1, so they are commute. Based on these observations, the claim in the lemma holds.
\end{proof}
Let us define a sequence of matrices $Q(t)$ such that each of these matrices is a product of $\tau N$ of $A_t$ matrices. Specifically,
$$Q(t)=\prod_{k=t}^{t+\tau |\mathcal{N}|-1}A_t.$$
In addition, define metrics to measure the similarity among the rows of a row stochastic matrix $X$ as follows
\begin{align}
\delta(X):&=\max_{j}\max_{i_1,i_2}|X_{i_1,j}-X_{i_2,j}| \nonumber \\
\lambda(X):&=1-\min_{i_1,i_2}\sum_{j}\min\{X_{i_1,j},X_{i_2,j}\}. \nonumber
\end{align}
There exists the following inequality between the defined quantities above.
\begin{lemma}
For any $m$ square row stochastic matrices $X(1),X(2),\cdots,X(m)$,
\begin{align}
\delta(X(1)X(2)\cdots X(m))\le \prod_{i=1}^{m}\lambda(X(i)). \nonumber
\end{align}
\label{lem: del_lam}
\end{lemma}
The proof of Lemma \ref{lem: del_lam} is provided in \cite{HajBar_58}. An important observation for $\lambda(X)$ is that if there exists a non-zero column for $X$, say all elements a bounded above from 0 by at least $\xi>0$, then $\lambda(X)\ge 1-\xi$. In addition, let $\bar{X}=\frac{1}{N}\mathbf{1}^{1}X$, we introduce the following lemma for the $\ell_2$ distance between $X$ and $\mathbf{1}\bar{X}$ in terms of $\delta(X)$.
\begin{lemma}
For any $X\in\mathbb{R}^{n\times n}$, we have
\begin{align}
\|X-\mathbf{1}\bar{X}\|\le n\delta(X) 
\end{align}
\label{lem: delta_2}
\end{lemma}
\begin{proof}
By definition of $\delta(X)$, we know that any two entries from the same column deviate from each other at most $\delta(X)$. Let $\bar{x}_i$ be the $i$-th element of vector $\bar{X}$. Also, we have
\begin{align}
|X_{i,j}-\bar{X}_j|=|X_{i,j}-\frac{1}{n}\sum_{k=1}^{n}X_{k,j}|&=|\frac{1}{n}\sum_{i=1}^{n}(X_{i,j}-X_{k,j})| \nonumber \\
&\le \frac{1}{n}\sum_{i=1}^{n}|X_{i,j}-X_{k,j}| \nonumber \\
&\le \delta(X)\nonumber. 
\end{align}
For $\ell_2$ norm, we have
\begin{align}
\|X-\mathbf{1}\bar{X}\|=\sqrt{\sum_{i,j}(X_{ij}-\bar{X}_j)^{2}}\le \sqrt{N^{2}(\delta(X))^{2}}\le n\delta(X). \nonumber 
\end{align}
\end{proof}
\begin{lemma}
For any $t$, matrix $Q(t)$ is a row stochastic matrix with $\lambda(Q(t))\le 1-\rho$, where $\rho:=\frac{1}{\beta^{\tau N}}>0$ is a strictly positive number.
\label{lem: lambda_Q}
\end{lemma}
\begin{proof}
Since $Q(t)$ is product of a series of row stochastic matrices $A_t$, the product $Q(t)$ is a row stochastic matrix.
From Lemma \ref{lem: A_low_bou}, we have 
$$\beta H[t]\le A_t.$$
As a result, we further have 
\begin{align}
\beta^{\tau |\mathcal{N}|}\prod_{k=t}^{t+\tau |\mathcal{N}|-1}H[t]\le\prod_{k=t}^{t+\tau |\mathcal{N}|-1}A_t =Q(t).
\label{eq: Q_t}
\end{align}
By Lemma \ref{lem: H_all_one}, there exists a column for the LHS of Eq.~ \eqref{eq: Q_t}, where all elements are strictly positive, more specifically, not smaller than $\rho$. As a result, we have $\lambda(Q(t))\ge 1-\rho$. In addition, this inequality holds for all $t\ge 0$.
\end{proof}

\subsection{Proof of Theorem \ref{thm: consensus}}
\begin{proof}
For simplicity, we use notation $\phi_{k}:=\phi(s_{k})$. For the parameter iterative update, we have 
\begin{align}
&w^{i}_{k+1} \nonumber \\
&=\sum_{j\in\mathcal{N}} A_k(i,j)\cdot w^{j}_k+\eta_k(r^{i}_{k+1}+\gamma \phi_{k+1}w^{i}_k-\phi_{k}w^{i}_k)\phi_{k} \nonumber \\
&=\sum_{j\in\mathcal{N}} A_k(i,j)\cdot w^{j}_k+\eta_k\phi_{k}(\gamma \phi_{k+1}-\phi_{k}) w^{i}_k+\eta_kr^{i}_{k+1}\phi_{k} .
\label{eq: w_i}
\end{align}
Then, we consider the dynamics of the vector $w_k:=(w^{1}_k,\cdots,w^{|\mathcal{N}|}_k)^{T}$ and we have
\begin{align}
&w_{k+1} \nonumber \\
&=A_k w_k+\eta_k w_k\phi_{k}(\gamma \phi_{k+1}-\phi_{k}) +\eta_k r_{k+1}\phi_{k} \nonumber \\
&=\prod_{l=0}^{k}A_lw_0+\sum_{l=0}^{k}\eta_l b_l \prod_{n=l+1}^{k}A_n w_l+\sum_{l=0}^{k}\eta_l\prod_{n=l+1}^{k}A_n C_l
\label{eq: w_d}
\end{align}
where $r_{k}=(r^{1}_{k},\cdots,r^{|\mathcal{N}|}_{k})^{T}$, in addition we defined $b_l=\phi_{l}(\gamma \phi_{l+1}-\phi_{l})$ and $C_l=r_{l+1}\phi_{l}$.

For $b_l$ and $C_l$, by Assumptions \ref{ass: r_bou} and \ref{ass: fea},we have
\begin{align}
|b_l|\le 2, \qquad \|C_l\|\le r_{\max}.  \nonumber
\end{align}

The dynamics of \eqref{eq: w_d} for $k$-th iteration is 
\begin{align}
w_k=\prod_{l=0}^{k-1}A_lw_0+\sum_{l=0}^{k-1}\eta_l b_l\prod_{n=l+1}^{k-1}A_n w_l+\sum_{l=0}^{k-1}\eta_l\prod_{n=l+1}^{k-1}A_n C_l. \nonumber
\end{align}
Let $\bar{w}_k=\frac{1}{N}\sum_{i\in\mathcal{N}} w^{i}_k$. Then, for the average dynamic, we have
\begin{align}
\bar{w}_k
&=p^{T}_{0,k-1}w_0+\sum_{l=0}^{k-1}\eta_l b_l p^{T}_{l+1,k-1} w_l+\sum_{l=0}^{k-1}\eta_l p^{T}_{l+1,k-1}C_l.
\label{eq: w_bar} 
\end{align}
where $p_{l,k}$ is defined as 
\begin{align}
p_{k,t}=\frac{1}{N}\mathbf{1}^{T}\prod_{l=k}^{t}A_l. \nonumber
\end{align}
Then, by Lemma \ref{lem: lambda_Q}, we have
\begin{align}
\delta(\prod_{l=k}^{t}A_l)&\le \prod_{l=k}^{t}\lambda (A_l) \nonumber \\
&\le \prod_{l=k}^{\lfloor \frac{t-k}{\tau |\mathcal{N}|}\rfloor}\lambda (Q_l) \nonumber \\
&\le (1-\rho)^{\lfloor \frac{t-k}{\tau |\mathcal{N}|}\rfloor}.
\end{align}
Recall that $\rho=\beta^{\tau |\mathcal{N}|}$.
For $\ell_2$ norm, by Lemma \ref{lem: delta_2}, we have
\begin{align}
\|\prod_{l=k}^{t}A_l-\mathbf{1}p^{T}_{k,l}\|&\le |\mathcal{N}|\delta(\prod_{l=k}^{t}A_l) \nonumber \\
&\le |\mathcal{N}|(1-\rho)^{\lfloor \frac{t-k}{\tau |\mathcal{N}|}\rfloor} \nonumber \\
&\le |\mathcal{N}|\zeta^{t-k},
\label{eq: 2_norm_dis}
\end{align}
where $\zeta$ is defined as $\zeta=(1-\rho)^{\lfloor \frac{1}{\tau |\mathcal{N}|}\rfloor}$.

For consensus error, we have
\begin{align}
&\TwoNorm{w_k-\bar{w}_k\mathbf{1}} \nonumber \\
=&\left\|\prod_{l=0}^{k-1}A_lw_0+\sum_{l=0}^{k-1}\eta_l b_l\prod_{n=l+1}^{k-1}A_n w_l+\sum_{l=0}^{k-1}\eta_l\prod_{n=l+1}^{k-1}A_n C_l-\mathbf{1}\cdot p^{T}_{0,k-1}w_0-\sum_{l=0}^{k-1}\eta_l b_l \mathbf{1}\cdot p^{T}_{l+1,k-1}w_l-\sum_{l=0}^{k-1}\eta_l\mathbf{1}\cdot p^{T}_{l+1,k-1}C_l\right\| \nonumber \\
=&\left\|\left(\prod_{l=0}^{k-1}A_l-\mathbf{1}\cdot p^{T}_{0,k-1}\right)w_0+\sum_{l=0}^{k-1}\eta_l b_l\left(\prod_{n=l+1}^{k-1}A_n -\mathbf{1}\cdot p^{T}_{l+1,k-1}\right)w_l+\sum_{l=0}^{k-1}\eta_l\left(\prod_{n=l+1}^{k-1}A_n-\mathbf{1}\cdot p^{T}_{l+1,k-1}\right) C_l \right\| \nonumber \\
\le&\left\|\left(\prod_{l=0}^{k-1}A_l-\mathbf{1}\cdot p^{T}_{0,k-1}\right)w_0\right\|+\left\|\sum_{l=0}^{k-1}\eta_l b_l\left(\prod_{n=l+1}^{k-1}A_n-\mathbf{1}\cdot p^{T}_{l+1,k-1}\right) w_l\right\|+\left\|\sum_{l=0}^{k-1}\eta_l\left(\prod_{n=l+1}^{k-1}A_n-\mathbf{1}\cdot p^{T}_{l+1,k-1}\right) C_l \right\| \nonumber \\
\le&\underbrace{\left\|\prod_{l=0}^{k-1}A_l-\mathbf{1}\cdot p^{T}_{0,k-1}\right\|}_\text{term (a)}\cdot\|w_0\|+\sum_{l=0}^{k-1}\eta_l \underbrace{\|\prod_{n=l+1}^{k-1}A_n-\mathbf{1}\cdot p^{T}_{l+1,k-1}\|}_\text{term (b)}\cdot\|w_l\|\cdot|b_l| \nonumber \\
&+\sum_{l=0}^{k-1}\eta_l \underbrace{\|\prod_{n=l+1}^{k-1}A_n-\mathbf{1}\cdot p^{T}_{l+1,k-1}\|}_\text{term (b)}\cdot \|C_l \|.
\label{eq: consensus_error}
\end{align}
By Lemma \ref{lem: delta_2} and Eq.~\eqref{eq: 2_norm_dis}, we have bounds on term (a) and (b) respectively as follows
\begin{align}
\left\|\prod_{l=0}^{k-1}A_l-\mathbf{1}\cdot p^{T}_{0,k-1}\right\|\le |\mathcal{N}|\cdot \zeta^{k-1} \nonumber  \\
\left\|\prod_{n=l+1}^{k-1}A_n-\mathbf{1}\cdot p^{T}_{l+1,k-1}\right\|\le |\mathcal{N}|\cdot \zeta^{k-l-2}.
\end{align}
Recall that the step size is $\eta_l=\frac{1}{l}$, as a result, we have
\begin{align}
\lim_{k\to\infty}\sum_{l=0}^{k-1}\eta_l\xi^{k-l-2}=0. \nonumber
\end{align}
Then, consensus error in Eq.~\eqref{eq: consensus_error}, we have
\begin{align}
&\TwoNorm{w_k-\bar{w}_k\mathbf{1}} \nonumber \\
\le&|\mathcal{N}|\cdot \zeta^{k-1}\cdot\|w_0\|+\sum_{l=0}^{k-1}\eta_l |\mathcal{N}|\cdot \zeta^{k-l-2}\cdot\|w_l\|\cdot|b_l|+\sum_{l=0}^{k-1}\eta_l |\mathcal{N}|\cdot \zeta^{k-l-2}\cdot \|C_l \|.
\end{align}
Since we used projected TD learning, we have $\|w_l\|\le\sqrt{|\mathcal{N}|}R$. For projection step, by Lemma \ref{lem: proj}, we have
\begin{align}
\|\Pi_{2,R}(w_k)-\overline{\Pi_{2,R}(w_k)}\mathbf{1}\|\le \TwoNorm{w_k-\bar{w}_k\mathbf{1}}.
\end{align}
As a result for the consensus error in Eq.~\eqref{eq: consensus_error}, we have
\begin{align}
\lim_{k\to\infty}\TwoNorm{w_k-\bar{w}_k\mathbf{1}}=0 \nonumber
\end{align}
\end{proof}

Here, we introduce a lemma that states the coordinate-wise projection operator has a contraction property.
\begin{lemma}
For the coordinate-wise projection operator $\Pi_{2,R}$, we have 
\begin{align}
\|\Pi_{2,R}(w)-\mathbf{1}\overline{\Pi_{2,R}(w)}\|\le\|w-\mathbf{1}\bar{w}\| \nonumber
\end{align}
\label{lem: proj}
\end{lemma}
\begin{proof}
For simplicity, we use $\Pi$ instead of $\Pi_{2,R}$.
Let $R^{+}$ be the index set $\{i|w^{i}>R\}$ and similarly $R^{-}=:\{i|w^{i}<-R\}$. Furthermore, let $\Delta w^{+}_i=w^{i}-R$ for $i\in R^{+}$ and $\Delta w^{-}_i=-R-w^{i}$ for $i\in R^{-}$. We note that both $\Delta w^{+}_i>0$ for $i\in R^{+}$ and $\Delta w^{-}_i>0$ for $i\in R^{-}$ by definition. Finally, $\Delta =-\sum_{i\in R^{+}}w^{+}_i+\sum_{i\in R^{-}}w^{-}_i$.

Then, the change of average before and after projection is $\frac{\Delta}{|\mathcal{N}|}$. In other words, we have
\begin{align}
\bar{w}+\frac{\Delta}{|\mathcal{N}|}=\overline{\Pi_{2,R}(w)}. \label{eq: ave_change}
\end{align}
For the square of consensus error after projection, we have 
\begin{align}
&\|\Pi(w)-\mathbf{1}\overline{\Pi(w)}\|^{2} \nonumber \\
=&\|\Pi(w)-\mathbf{1}(\bar{w}+\frac{\Delta}{|\mathcal{N}|})\|^{2} \nonumber \\
=&\underbrace{\sum_{i\in \mathcal{N}\setminus (R^{+}\bigcup R^{-})}(w^{i}-\bar{w}-\frac{\Delta}{|\mathcal{N}|})^{2}}_{\text{term (a)}}
+\underbrace{\sum_{i\in R^{+}}(w^{i}-\Delta w^{+}_i-\bar{w}-\frac{\Delta}{|\mathcal{N}|})^{2}}_{\text{term (b)}}
+\underbrace{\sum_{i\in R^{-}}(w^{i}+\Delta w^{-}_i-\bar{w}-\frac{\Delta}{|\mathcal{N}|})^{2}}_{\text{term (c)}} \nonumber 
\end{align}
For each of the term above, we have
\begin{align}
\text{term (a)}&=\sum_{i\in \mathcal{N}\setminus (R^{+}\bigcup R^{-})}(w^{i}-\bar{w})^{2}-2\sum_{i\in \mathcal{N}\setminus (R^{+}\bigcup R^{-})}(w^{i}-\bar{w})\frac{\Delta}{|\mathcal{N}|}+\sum_{i\in \mathcal{N}\setminus (R^{+}\bigcup R^{-})}\frac{\Delta^{2}}{|\mathcal{N}|^{2}} \\
\text{term (b)}&=\sum_{i\in R^{+}}(w^{i}-\bar{w})^{2}-2\sum_{i\in R^{+}}(w^{i}-\bar{w})(\Delta w^{+}_i+\frac{\Delta}{|\mathcal{N}|})+\sum_{i\in R^{+}}(\Delta w^{+}_i+\frac{\Delta}{|\mathcal{N}|})^{2} \\
\text{term (c)}&=\sum_{i\in R^{-}}(w^{i}-\bar{w})^{2}+2\sum_{i\in R^{-}}(w^{i}-\bar{w})(\Delta w^{-}_i-\frac{\Delta}{|\mathcal{N}|})+\sum_{i\in R^{-}}(\Delta w^{-}_i-\frac{\Delta}{|\mathcal{N}|})^{2}
\end{align}
The first terms in the term (a), (b) and (c) are the original consensus error. We now add the cross terms in term (a), (b) and (c). We have
\begin{align}
&-2\sum_{i\in \mathcal{N}\setminus (R^{+}\bigcup R^{-})}(w^{i}-\bar{w})\frac{\Delta}{|\mathcal{N}|}-2\sum_{i\in R^{+}}(w^{i}-\bar{w})(\Delta w^{+}_i+\frac{\Delta}{|\mathcal{N}|})+2\sum_{i\in R^{-}}(w^{i}-\bar{w})(\Delta w^{-}_i-\frac{\Delta}{|\mathcal{N}|}) \nonumber \\
=&-2\sum_{i\in R^{+}}(w^{i}-\bar{w})\Delta w^{+}_i+2\sum_{i\in R^{-}}(w^{i}-\bar{w})\Delta w^{-}_i.
\end{align}
We note that both terms above are negative. Similarly, we add the squared terms, then we have
\begin{align}
&\sum_{i\in \mathcal{N}\setminus (R^{+}\bigcup R^{-})}\frac{\Delta^{2}}{|\mathcal{N}|^{2}}+\sum_{i\in R^{+}}(\Delta w^{+}_i+\frac{\Delta}{|\mathcal{N}|})^{2}+\sum_{i\in R^{-}}(\Delta w^{-}_i-\frac{\Delta}{|\mathcal{N}|})^{2} \nonumber \\
=&\frac{\Delta^{2}}{|\mathcal{N}|}+2\frac{\Delta}{|\mathcal{N}|}(\sum_{i\in R^{+}}\Delta w^{+}_i-\sum_{i\in R^{-}}\Delta w^{-}_i)+\sum_{i\in R^{+}}(\Delta w^{+}_i)^{2}+\sum_{i\in R^{-}}(\Delta w^{-}_i)^{2} \nonumber \\
=&-\frac{\Delta^{2}}{|\mathcal{N}|}+\sum_{i\in R^{+}}(\Delta w^{+}_i)^{2}+\sum_{i\in R^{-}}(\Delta w^{-}_i)^{2}.
\end{align}
As a result, we have
\begin{align}
&\|\Pi(w)-\mathbf{1}\overline{\Pi(w)}\|^{2} \nonumber \\
=&\|w-\mathbf{1}\bar{w}\|^{2}-2\sum_{i\in R^{+}}(w^{i}-\bar{w})\Delta w^{+}_i+2\sum_{i\in R^{-}}(w^{i}-\bar{w})\Delta w^{-}_i-\frac{\Delta^{2}}{|\mathcal{N}|}+\sum_{i\in R^{+}}(\Delta w^{+}_i)^{2}+\sum_{i\in R^{-}}(\Delta w^{-}_i)^{2}.
\label{eq: proj_con_err}
\end{align}
We consider 3 cases: (1) $-R\le\bar{w}\le R$, (2) $R<\bar{w}$ and (3) $\bar{w}< -R$.

For case (1), it's easy to see that $w^{i}-\bar{w}\ge\Delta w^{+}_i$ for $i\in R^{+}$ and $w^{i}-\bar{w}\le -\Delta w^{-}_i$ for $i\in R^{-}$. As a result, the projection is indeed contraction.

For case (2), which implies $\Delta<0$, we have $R< \bar{w}\le R-\frac{\Delta}{|\mathcal{N}|}$. The RHS is because after projection, the average needs to be smaller or equal to $R$. Then, we have
\begin{align}
\sum_{i\in R^{+}}(w^{i}-\bar{w})\Delta w^{+}_i &\ge \sum_{i\in R^{+}}(\Delta w^{+}_i+\frac{\Delta}{|\mathcal{N}|})\Delta w^{+}_i \nonumber \\
&=\sum_{i\in R^{+}}(\Delta w^{+}_i)^{2}+\frac{\Delta}{|\mathcal{N}|}\sum_{i\in R^{+}}\Delta w^{+}_i 
\end{align}
And we also have 
\begin{align}
\sum_{i\in R^{-}}(w^{i}-\bar{w})\Delta w^{-}_i&\le \sum_{i\in R^{-}}(w^{i}-R)\Delta w^{-}_i \nonumber \\
&=  \sum_{i\in R^{-}}(w^{i}+R-2R)\Delta w^{-}_i \nonumber \\
&= \sum_{i\in R^{-}}(-\Delta w^{-}_i-2R)\Delta w^{-}_i \nonumber \\
&=-\sum_{i\in R^{-}}(\Delta w^{-}_i)^{2}-2R\sum_{i\in R^{-}} \Delta w^{-}_i. 
\end{align}
With the above two inequalities, we have the contraction. Similarly, it holds for case (3). Therefore, the proof is complete.
\end{proof}

\end{document}